\documentclass[journal]{IEEEtran}

\usepackage{multicol}

\ifCLASSINFOpdf
\usepackage[pdftex]{graphicx}
\usepackage{dblfloatfix}

\graphicspath{{../pdf/}{../jpeg/}}
\DeclareGraphicsExtensions{.pdf,.jpeg,.png}
\else
\fi
\usepackage[table]{xcolor}
\usepackage{tikz, cite}
\usetikzlibrary{decorations.pathreplacing,calligraphy}
\usepackage{pgfplots}
\pgfplotsset{width=8cm,compat=1.9}
\usepgfplotslibrary{external}
\usepackage[hypcap=false]{caption}
\usepackage{subcaption}
\usepackage{url}
\usepackage{hyperref}
\usepackage{amsmath}
\usepackage{amssymb}
\usepackage{setspace,csquotes}
\usepackage{enumerate}
\allowdisplaybreaks
\usepackage[square,numbers]{natbib}
\usepackage{cuted}
\usepackage{flushend}
\usepackage{nopageno}
\newcommand\blfootnote[1]{%
  \begingroup
  \renewcommand\thefootnote{}\footnote{#1}%
  \addtocounter{footnote}{-1}%
  \endgroup
}

\usepackage{url}

\usepackage{amsthm}
\usepackage{dblfloatfix}
\newtheorem{prop}{Proposition}
\newtheorem{theorem}{Theorem}
\newtheorem{lemma}{Lemma}
\newtheorem{deftn}{Definition}
\newtheorem{corol}{Corollary}

\begin{document}
%
\title{MiSTA: An Age-Optimized Slotted ALOHA Protocol}
%
%
%

\author{Mutlu Ahmetoglu, Orhan Tahir Yavascan and Elif Uysal \\
Dept. of Electrical and Electronics Engineering, METU, 06800, Ankara, Turkey \\
  \{mutlu.ahmetoglu,orhan.yavascan,uelif\}@metu.edu.tr
}

%


\setlength{\belowdisplayskip}{2pt} \setlength{\belowdisplayshortskip}{2pt}
\setlength{\abovedisplayskip}{2pt} \setlength{\abovedisplayshortskip}{2pt}


    \vspace*{15pt}
    {\let\newpage\relax\maketitle}
\thispagestyle{empty} 
\vspace{-2cm}
\begin{abstract}
\blfootnote{This work was supported in part by TUBITAK grant 117E215, and by Huawei. A preliminary version was reported in the BlackSeaCom 2021 conference.\cite{MutluBlackSeaCom2021}} We introduce Mini Slotted Threshold ALOHA (MiSTA), a slotted ALOHA modification designed to minimize the network-wide time average Age of Information (AoI). In MiSTA, sources whose ages are  below a certain threshold stay silent. When a node with age above the threshold has data to send, it becomes active in the next time frame with a certain probability. The active node first transmits a short control sequence in a mini-slot ahead of actual data transmission, and if collision is sensed, it backs off with a certain probability. We derive the steady state distribution of the number of active sources and analyze its limiting behaviour. We show that MiSTA probabilistically converges to a \enquote{thinned} slotted ALOHA, where the number of active users at steady state adjusts to optimize age. With an optimal selection of parameters, MiSTA achieves an AoI scaling with the number of sources, $n$, as $0.9641n$, which is an improvement over the Threshold ALOHA policy proposed earlier (for which the lowest possible scaling is $1.4169n$). While achieving this reduction in age, MiSTA also increases achievable throughput to approximately $53\%$, from the $37\%$ achievable by Threshold ALOHA and regular slotted ALOHA.
\end{abstract}


%
\IEEEpeerreviewmaketitle

\vspace{-0.3cm}
\section{Introduction}
	%
	%
	%
	%
\par Proliferation of IoT, autonomous mobility and remote monitoring applications is influencing the redesign of wireless access protocols to better cater for Machine-Type Communications (MTC). The information timeliness requirements in massive deployments of nodes generating short, sporadic data packets are not captured adequately by conventional protocol principles, based on classical performance metrics. The Age of Information (AoI) metric was proposed to more directly capture the timeliness of flows in such applications \cite{Altman2010}. Optimization according to AoI leads to non-conventional (and sometimes counter-intuitive) service and scheduling policies \cite{kaul2012real,sun2017update}. Age optimization has been studied under many network and service models, including random access schemes, leading to suggested innovations at various networking layers~\cite{Yates2021Survey,Mankar,Munari2021,Grybosi,Yu2021,Han2021,YatesKaul5,Inoue,Yates12}. In massive MTC scenarios where nodes are to occasionally send typically short status update packets to a common access point over a shared channel, a simple slotted ALOHA type policy is preferable to scheduled access or even to CSMA-type random access due to the low tolerance for overhead. Accordingly, in recent literature \cite{munari2021information,Saha,Orhan,chen2020age,shirin,doga,jiang2018timely} various slotted ALOHA variants have been studied with a focus on the AoI performance. We now briefly review this recent literature to put the present paper into context.

\par In slotted ALOHA, sources that wish to send data transmit with a certain probability $\tau$ in each slot. The selection of $\tau$ to minimize time average age was studied in \cite{yates2017status}. In \cite{doga} it was proposed to activate the transmission with probability $\tau$ only when a user's age reaches a certain threshold $\gamma$, resulting in a policy termed \enquote{Threshold ALOHA (TA)}. The parameters $\tau$ and $\gamma$ were jointly optimized in \cite{OrhanGlobecom,Orhan}. Of course, this policy relies on the availability of feedback to the sender upon a successful transmission, which it then uses to update its age. The benefit of this one bit feedback per successful transmission was shown to be significant, in \cite{Orhan}: the average AoI scales with the network size $n$ as 1.4169$n$, in contrast to the best achievable age scaling $en$ with slotted ALOHA. Moreover, there is nearly no loss in throughput with respect to slotted ALOHA. 

\par The work reported in \cite{shirin} combines the age-based randomized transmission approach with a more detailed dynamic optimization. The SAT policy proposed therein obtains an age-based thinning that is similar to TA except it uses dynamic thresholds resembling a stabilized slotted ALOHA mechanism: sources make decentralized estimations of their ages and compute an age-based transmission threshold, above which they transmit with constant probability. The asymptotic scaling of average AoI with network size, achieved by SAT is slightly better than that achieved by TA: $en/2\approx1.359n$.

The work to be presented in this paper is also related to the studies in \cite{maatouk2019minimizing} and \cite{kadotaInfo2021} with respect to its use of reservation slots: In \cite{maatouk2019minimizing}, a standard CSMA model with stochastic arrivals was studied for AoI optimization, reaching the conclusion that standard CSMA does not cater well to age performance. The combined analytical and experimental study reported in \cite{kadotaInfo2021} also focused on CSMA with stochastic arrivals to nodes at the beginning of each frame, which is divided into minislots. In \cite{kadotaInfo2021}, nodes that capture the channel use a certain number of mini slots combined. Sources do not track their ages, and thus make age-oblivious decisions. The analytically derived   average AoI achieved by this policy was supported by experimental results obtained on a Software Defined Radio test-bed implementation.

\par The main contributions of our paper are the following:
	\begin{itemize}
	\item We introduce Mini Slotted Threshold ALOHA (MiSTA) as the following modification of Threshold ALOHA: at the beginning of a slot, each active source (that is, each backlogged node with age exceeding $\Gamma$) sends a short beacon signal with probability $\tau_1$. If the source receives a collision feedback at the end of this mini-slot, it will back-off with probability $1-\tau_2$. With probability $\tau_2$, it will go ahead and transmit a data packet during the remainder of the slot.
	\item We derive the steady state distribution of the age vector achieved by MiSTA under the setting that sources generate data \enquote{at will}. We use this solution to compute the steady-state distribution of the number of active users for any given network size, $n$ (Lemma \ref{lemma:truncated}).  
	\item
	We then analyze the asymptotic behaviour of the system in the large networks. In particular, we establish the independence of the number of active sources and the age of an individual source in the limiting case (Corollary \ref{cor:pivot}). This observation is used as the key ingredient in determining the probability of a successful transmission, $q_o$, in steady state.
	\item 
	We show that as the number of sources increases to infinity, the system converges to a \textit{thinned} slotted ALOHA network (i.e. a slotted ALOHA network with fewer nodes) (Theorem \ref{thm:1} and \ref{thm:2}). This resembles the results achieved by Rivest's stabilized slotted ALOHA, or the age-thinning policy introduced in \cite{shirin}.
	\item 
	We derive an expression for the steady-state time average AoI in terms of $n$, $\Gamma$, $\tau_1$ and $\tau_2$. With the optimal choice of parameters, this expression gives the time average AoI of $0.9641n$ for the system (Theorem \ref{thm:ageResult}). The corresponding value is $1.4169n$ for the Threshold ALOHA policy, which implies that MiSTA reduces the time average AoI by almost a third. Meanwhile, MiSTA achieves approx. $53\%$ throughput, compared to the $37\%$ achievable by Threshold ALOHA or plain slotted ALOHA.
	\item
	We show that the time used for minislots results in no loss; on the contrary, MiSTA achieves a net gain in spectral efficiency over TA and ordinary slotted ALOHA. Finally we briefly introduce an extended version of MiSTA, with multiple mini slots per frame. We numerically compare the performance of this policy with the one proposed in \cite{kadotaInfo2021}.
	\end{itemize}
\par The rest of the paper is organized as follows: In Section \ref{sect:systemmodel} the system model is presented. Our proposed policy, Mini Slotted Threshold ALOHA (MiSTA) is defined in Section \ref{sect:minislot}. In Section \ref{sect:steady} the steady state distribution of age under MiSTA for any network size is derived. Section \ref{sect:PivtChn} and Section \ref{sect:rootAnalysis} analyze further possibilities for the asymptotic steady state behaviour of the system. Section \ref{sect:ageCalc} derives age-optimal policy parameters. Section \ref{sect:lowerbound} provides an upper bound for the throughput and a lower bound for the time average AoI attained by MiSTA. Section \ref{sect:spectral} analyzes the spectral efficiency of MiSTA. Section \ref{sect:numerical} provides a numerical study to further illustrate the performance of this policy, and contrast it with  related policies from the literature. We conclude in section \ref{sect:conclusion} with a discussion of possible extensions.   

\section{System Model} \label{sect:systemmodel}
We consider a wireless random access channel with $n$ sources and a single access point (AP). The sources access the AP on a shared channel to send time sensitive status update data to their respective destinations. Any delays from the AP to the destinations is ignored. Time is slotted and all nodes are assumed to be synchronized. Two or more transmissions that occur in the same slot result in a collision where no packet is successfully decoded. We will adopt the \enquote{generate-at-will} model~\cite{SunIT2017} where sources take fresh samples of data when they decide to transmit instead of re-transmitting data that failed to be successfully transmitted earlier. 

We define $A_i[t]$ as the Age of Information (AoI) of source $i\in\{1,\ldots,n\}$ at the time slot $t$. By definition, $A_i[t]$ is the time elapsed since the generation of the most recent successful data transmitted by source $i$. In accordance with the generate-at-will model assumption, all received packets reduce the age of the corresponding flow to $1$. Hence, $A_i[t]$ is equal to the the number of slots that have passed by time $t$ since the latest successful transmission by source $i$, plus 1. Adopting the setting of Threshold ALOHA \cite{Orhan}, we assume one bit feedback per successful transmission, such that the AP informs the successful source in the event of successful decoding. Consequently, the age process evolves as
\begin{equation}
\begin{aligned}
    A_{i}[t] = \left\{
    \begin{array}
    {ll}1,  & \begin{aligned}&\text {src. } i\text { transmits successfully} \\&\text{at} \text{ time slot } t-1  \end{aligned}\\
    A_{i}[t-1]+1, & \text {otherwise }
    \end{array}
    \right.
    \end{aligned}
\end{equation}
The long term average AoI of source $i$ is defined as
\begin{equation} \label{eq:2}
\Delta_{i}=\lim_{T \rightarrow \infty} \frac{1}{T} \sum_{t=0}^{T-1} A_{i}[t]
\end{equation}
whenever the limit exists.
In the analysis of Threshold ALOHA, the system age vector was constructed with ages of individual sources as
\begin{equation}
    \mathbf{A}[t] \triangleq\left\langle A_{1}[t] \quad A_{2}[t] \quad \ldots \quad A_{n}[t]\right\rangle
\end{equation}
In \cite{doga}, it was shown that the evolution of this vector follows a Markov Chain and that a truncated version of this Markov Chain suffices for the analysis of steady-state long term average age in the original model. In this finite state truncated Markov Chain, the source ages are truncated at $\Gamma$ since the sources behaviour does not change after it becomes active. The same truncation argument is valid for the Markov Chain analysis for MiSTA, as will be shown in the next section. Moreover, the Markov Chain is ergodic, and using the symmetry between the sources the time average AoI in the network can written as the limiting expected AoI of a single node:
\begin{equation}
    \Delta_i=\lim_{t \to \infty} \mathbb{E} \left[A_{i}[t]\right]
\end{equation}

\section{Mini Slotted Threshold ALOHA}  \label{sect:minislot}
\par MiSTA is a modification of Threshold ALOHA where a mini slot is prepended to each slot. The remainder of the slot is of sufficient length for one packet transmission.
\begin{figure}[ht] 
\centering
\includegraphics{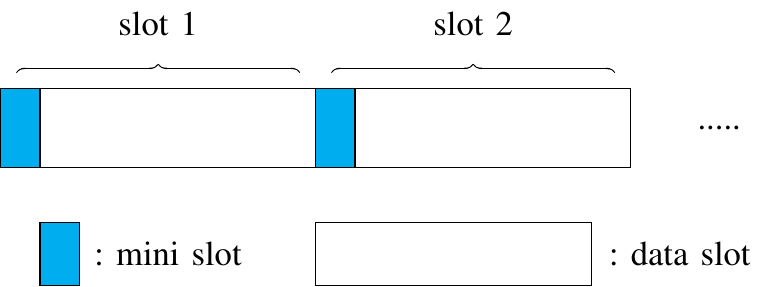}
\caption{The slot structure for Mini Slotted Threshold ALOHA}
\label{fig:slots}
\end{figure}
\par The mini slots are for channel sensing.  Before each mini slot, each active source independently decides to become an attempter with probability $\tau_1$. During the mini slot, attempters make a short transmission (e.g., consisting of a flow ID). If a source receives a success feedback from the AP at the end of the mini slot, it infers that it was the sole attempter and readies a fresh data sample for transmission in the data slot. On the other hand, if it does not receive success feedback for its transmission in the mini slot, it deduces that there were multiple contenders. Each contender independently decides to attempt again and send data in the data slot with probability $\tau_2\leq \tau_1$. Hence, MiSTA, in relation to TA, gives nodes a second chance to avoid collision. 
\subsection{Steady State Distribution of the Markov Chain Model} \label{sect:steady}
In this subsection, we analyze the truncated Markov Chain that captures the evolution of the age vector in a mini slotted environment. Using the steady state probabilities of the recurrent states, we will derive the distribution of the number of active users, $m$. First, we define the truncated system age vector as
\begin{equation}
\mathbf{A}^{\Gamma}[t] \triangleq
\left\langle A_{1}^{\Gamma}[t] \quad A_{2}^{\Gamma}[t] \quad \ldots \quad A_{n}^{\Gamma}[t]\right\rangle
\end{equation}
where $A_{i}^{\Gamma}[t] \in \{1,2,\ldots,\Gamma\}$ is the Aol of source {\it i} at time {\it t} $\in \mathbb{Z}^+$ truncated at $\Gamma$. $A_{i}^{\Gamma}[t]$ is equal to the actual age if the age is below $\Gamma$, and is set equal to $\Gamma$ otherwise. Therefore, the $A_{i}^{\Gamma}[t]$ expression evolves as:
\begin{equation}
\begin{aligned}
A_{i}^{\Gamma}[t+1]=\left\{\begin{array}{ll}1, & \begin{aligned}&\text {src. } i \text { updates } \\&\text{at } \text{time $t-1$, }\end{aligned} \\ \min \left\{A_{i}^{\Gamma}[t]+1,\right.  \Gamma\}, &\text { otherwise. }\end{array}\right.
\end{aligned}
\end{equation}
$\{\textbf{A}^{\Gamma}[t], t\geq 1\}$ is a Markov Chain with a finite state space and a unique steady state distribution, as can be trivially shown following the corresponding argument in \cite{doga}: The truncated system is a homogeneous Markov Chain because the evolution of the system state depends on the past only through the current vector of ages, i.e. the current state. The transition probabilities depend only on the age vector. For a finite network size, $n$, the state space is finite as the values taken by the elements of the state vector are chosen from the set of integers $\{1,2,\ldots,\Gamma\}$. Note that, the state with all entries equal to $\Gamma$ is accessible from any state in a finite number of transitions. From any given state, it takes a finite number of slots without a successful transmission to go to this particular state, so that even the source with the smallest age becomes active. Then, the state with all entries equal to $\Gamma$ is accessible in finite time from all recurrent states including itself. Consequently, the MC forms a single, aperiodic recurrent class, and consequently it has a unique steady state distribution. 
\par The truncated MC model of \cite{doga} was further  studied in \cite{Orhan}. Similar to \cite{Orhan}, we argue that under the MiSTA policy, states with repeating age values smaller than $\Gamma$ (multiple nodes with the same age) are transient. This is because states with a repeating below threshold age value $s$ suggest a simultaneous transmission for these states $s$ slots ago, which is not feasible since only one source can successfully transmit in a slot.
\begin{prop}\cite{Orhan}
\label{prop:recurrentClass}
	For distinct indices i and j, if a state $\left\langle s_{1} \quad s_{2} \quad \ldots \quad s_{n}\right\rangle$ in the truncated MC $\{\textbf{A}^{\Gamma}[t], t\geq 1\}$ is recurrent, then $s_i = s_j$ if and only if $s_i = s_j = \Gamma$.
\end{prop}
We note that all the nodes that satisfy the condition in Prop. \ref{prop:recurrentClass} can be reached from the state with all entries equal to $\Gamma$, which is known to be in the only recurrent class of the Markov Chain, therefore, all such states are recurrent.
\par Next, in order to simplify the identification of recurrent states, we adopt the following notation of \cite{Orhan} whereby the \textit{type} of a recurrent state is written as:
\begin{equation}
    T\langle s_{1} \quad s_{2} \quad \ldots \quad s_{n}\rangle = (M,\{u_{1}, u_{2}, \ldots, u_{n-M} \}),
\end{equation}
In this notation, $M$ is the number of active sources (i.e. sources with the truncated age of $\Gamma$), and the set $\{u_{1}, u_{2}, \ldots, u_{n-M} \}$ is the set of age values in the state that correspond to the passive sources (i.e. sources with an age smaller than $\Gamma$). The motivation behind this notation stems from the following proposition that simplifies the steady state analysis of MiSTA policy to a great extent:
\begin{prop} \label{Prop:3}
    States with the same types have equal steady state probabilities.
\end{prop}
\begin{proof}
States with the same types can be reached from one another by altering the order of the state entries. Due to the symmetry between the users, modifying the order of the state entries does not change the steady state probability of a state.
\end{proof}
\par Next, we show that the distribution of $M$ in steady state can be derived from the distribution of the Markov Chain. We validate this in Lemma \ref{lemma:truncated} by showing that the number $M$ for a state, which corresponds to the number of active sources pertaining to it, is what determines the steady state probability of that state.
\begin{lemma} \label{lemma:truncated}
The following are valid for the truncated Markov Chain $\{\textbf{A}^{\Gamma}[t], t\geq 1\}$:
\begin{enumerate}[i)]
  \item M, which is the number of active sources given a state vector $\left\langle s_{1} \quad s_{2} \quad \ldots \quad s_{n}\right\rangle$, is what solely determines the steady state probability of the state vector.
  \item Let $P_m$ be the total steady state probability of states with $m$ active users. Then the $\frac{P_m}{P_{m-1}}$ expression is found\footnote{The expression is:\\ \centerline{$\frac{\big(1/\tau_1 - (M-1)[(1-\tau_2)(1-\tau_1)^{M-2}+\tau_2(1-\tau_1\tau_2)^{M-2}]\big)(n-m+1)} {\big((1-\tau_2)(1-\tau_1)^{M-1}+\tau_2(1-\tau_1\tau_2)^{M-1}\big) m(\Gamma-1-n+m)}$.}}.
\end{enumerate}
\end{lemma}
\begin{proof}
For the proof of Lemma 1, we will define 4 types of state vectors. Let the type 1 be defined as $\mathcal{T}_1 \triangleq (M,\{u_{1}, u_{2}, \ldots, u_{n-M} \})$, where $M$ is the number of active sources as before and the set $\{u_{1}, u_{2}, \ldots, u_{n-M} \}$ contains no entry equal to 1. This means that there has not been a successful transmission in the previous slot. Note that there can be at most one entry equal to 1 in the recurrent states, which occurs in the case of successful transmission in the previous slot. Then, the state in the previous slot can be one of the two types defined as:
\begin{itemize}
	\item 
	$\mathcal{T}_2 \triangleq (M,\{u_{1}-1, u_{2}-1, \ldots, u_{n-M}-1 \})$
	\item
	$\mathcal{T}_3 \triangleq (M-1,\{\Gamma-1, u_{1}-1, u_{2}-1, \ldots, u_{n-M}-1 \})$
\end{itemize}
If there was no passive source with the age equal to $\Gamma -1$ in the previous slot, then there must have been $M$ active slots in the previous slot, which is the case represented by type $\mathcal{T}_2$. On the other hand, if there was a passive source with the age equal to $\Gamma -1$ in the previous slot, then there must have been $M-1$ active slots in the previous slot since this passive source turns active in this slot, which is the case represented by type $\mathcal{T}_3$. 
The fourth type is defined as $\mathcal{T}_0 \triangleq   (M-1,\{u_{1}, u_{2}, \ldots, u_{n-M},1 \})$ and it represents the case where the previous slot is one of the types $\mathcal{T}_2$ and $\mathcal{T}_3$ and there is a successful transmission in that slot by one of the nodes. As a result of the successful transmission, one of the $\Gamma$ entries will be replaced by 1 leading to the difference between $\mathcal{T}_0$ and $\mathcal{T}_1$.
When there are $M$ active sources at a time slot, the probability of there being no collision in the mini slot and a sole node that captures the data slot is $M\tau_1 (1-\tau_1)^{M-1}$. If there are $j$($j>1$) sources that attempt a transmission in the mini slot, these $j$ sources shall attempt another transmission in the data slot with probability $\tau_2$ each. Then, the probability of having a collision in the mini slot and a successful transmission in the data slot is:
\begin{equation}
\begin{aligned}
    &\textrm{P}_{c,s} = \sum_{j=2}^{M} \binom{M}{j}\tau_1^j(1-\tau_1)^{M-j} j\tau_2(1-\tau_2)^{j-1} \\
	&\stackrel{(a)}{=} M\tau_1\tau_2 \sum_{j=2}^{M} \binom{M-1}{j-1}\tau_1^{j-1}(1-\tau_2)^{j-1}(1-\tau_1)^{M-j} \\
	&\stackrel{(b)}{=} M\tau_1\tau_2 \left((\tau_1 (1-\tau_2) + (1-\tau_1))^{M-1} - (1-\tau_1)^{M-1} \right) \\
	&= M\tau_1\tau_2\left((1-\tau_1\tau_2)^{M-1}-(1-\tau_1)^{M-1}\right)
\end{aligned}
\end{equation}
where (a) follows from $j\binom{M}{j} = M \binom{M-1}{j-1}$ and (b) follows from the binomial sum. We further note that if the previous slot is one of the types $\mathcal{T}_2$ and $\mathcal{T}_3$, then the present slot can be one of the two types $\mathcal{T}_0$ and $\mathcal{T}_1$. A type $\mathcal{T}_2$ state precedes a type $\mathcal{T}_1$ state with probability $1-M\tau_1(1-\tau_1)^{M-1}-M\tau_1\tau_2[(1-\tau_1\tau_2)^{M-1}-(1-\tau_1)^{M-1}]$ which is the failure probability in a slot since this case corresponds to all sources failing to transmit. Likewise, a type $\mathcal{T}_2$ state precedes one of the $M$ possible type $\mathcal{T}_0$ states with probability $\tau_1(1-\tau_1)^{M-1}+\tau_1\tau_2[(1-\tau_1\tau_2)^{M-1}-(1-\tau_1)^{M-1}]$ which is the successful transmission probability for a single active source in a slot. If we define $\pi_{\mathcal{T}_i}$ to be the steady state probability of a state of type $\mathcal{T}_i$, we can use the above methodology to derive the probabilities of all the transitions between types of states in order to write the following equations:
\begin{equation} \label{eq:10} \hspace{-0.2cm} \small
\begin{aligned}
\pi_{\mathcal{T}_1} &= \pi_{\mathcal{T}_2}\big(1-M\tau_1[(1-\tau_2)(1-\tau_1)^{M-1}+\tau_2(1-\tau_1\tau_2)^{M-1}]\big)\\+  \pi_{\mathcal{T}_3}&M\big(1-(M-1)\tau_1[(1-\tau_2)(1-\tau_1)^{M-2}+\tau_2(1-\tau_1\tau_2)^{M-2}]\big)
\end{aligned}
\end{equation}
\begin{equation} \label{eq:11} \small
\begin{aligned}
&\pi_{\mathcal{T}_0} = \pi_{\mathcal{T}_2}\tau_1\big((1-\tau_2)(1-\tau_1)^{M-1}+\tau_2(1-\tau_1\tau_2)^{M-1}\big)\\& +  \pi_{\mathcal{T}_3}(M-1)\tau_1\big((1-\tau_2)(1-\tau_1)^{M-2}+\tau_2(1-\tau_1\tau_2)^{M-2}\big)
\end{aligned}
\end{equation}
As (\ref{eq:10}) and (\ref{eq:11}) represent all possible ways of transition between the recurrent states and since the truncated Markov Chain has a unique steady state distribution, this set of equations fully specifies the steady state probabilities. Part (\textit{i}) of the Lemma can be proven here by assigning the probability $\pi_M$ to the states with $M$ active sources. To verify this, we make the substitutions $\pi_{\mathcal{T}_1}=\pi_{\mathcal{T}_2}=\pi_M$ and $\pi_{\mathcal{T}_0}=\pi_{\mathcal{T}_3}=\pi_{M-1}$. 
Then, \eqref{eq:10} and \eqref{eq:11} both yield the same expression when simplified:
\begin{equation} \label{eq:piM} \small
\hspace{-0.4cm}
\frac{\pi_{M}}{\pi_{M-1}} = \frac{1/\tau_1 - (M-1)\big((1-\tau_2)(1-\tau_1)^{M-2}+\tau_2(1-\tau_1\tau_2)^{M-2}\big)} {(1-\tau_2)(1-\tau_1)^{M-1}+\tau_2(1-\tau_1\tau_2)^{M-1}}.
\end{equation}
\par This proves part \textit{(i)}. We will use this to find the total probability of having $m$ active users in steady state. Due to part \textit{(i)}, all recurrent states with $m$ active sources have the same steady state probability, which is $\pi_m$. The total number of recurrent states with $m$ active sources that obey Prop. \ref{prop:recurrentClass} is calculated as:
	\begin{equation}
	N_m = \binom{n}{m}\frac{(\Gamma-1)!}{(\Gamma-n-1+m)!}
	\end{equation}
Therefore, $P_m = N_{m}\pi_m$ and $\frac{P_m}{P_{m-1}} = \frac{N_m}{N_{m-1}}\frac{\pi_m}{\pi_{m-1}}$, from which \textit{(ii)} immediately follows.
\end{proof}
\subsection{Steady State Analysis Using the Pivot Node Technique} \label{sect:PivtChn}
This section will bridge the essential steady state analysis of the previous section with the more detailed analysis and derivation of optimal policy parameters in the next two, through proving certain technical results that culminate in Corollary 1.  We will use the \enquote{pivot source} technique introduced in \cite{Orhan}, where an arbitary node is chosen as pivot, without loss of generality, as the sources are symmetric. We will analyze the system through the states of this pivot source by modifying the truncated Markov Chain in the previous subsection $\{\textbf{A}^{\Gamma}[t], t\geq 1\}$, and obtaining a \textit{pivoted Markov Chain} $ \{\textbf{P}^{\Gamma}[t], t\geq 1\}$. In this pivoted Markov Chain, all the source ages except the pivot source will be truncated at $\Gamma$, but the age of the pivot source is kept as it is. As in the previous subsection, we will define the \textit{type} of a state in the pivoted Markov Chain in order to extend the results of Lemma \ref{lemma:truncated} as:
\begin{equation}
    \textrm{T}^{\textbf{\textrm{P}}}\langle S^{\textbf{P}}\rangle \triangleq (s,M,\{u_{1}, u_{2}, \ldots, u_{n-M-1} \})
\end{equation}
According to this new notation, $s \in \mathbb{Z}^+$ is the state of the pivot source,  M is the number of active sources (i.e. the sources with the entry $\Gamma$) excluding the pivot, and the set $\{u_{1}, u_{2}, \ldots, u_{n-M} \}$ is the set of entries belonging to passive sources (i.e. sources with entry smaller than $\Gamma$), again excluding the pivot source. Similar to the truncated Markov Chain analysis, we will refer to such a state as \textit{type $M$-state} where it is clear from the context.
\begin{prop} \label{prop:pivot}
\begin{enumerate} [(i)]
    \item $\textbf{P}^{\Gamma}$ has a unique steady state distribution.
    \item A type-$m$ state in $\textbf{P}^{\Gamma}$ has a steady state probability equal to $\pi_m$, obeying (\ref{eq:piM}), given that $s\in\{1,2,\ldots,\Gamma-1\}$.
\end{enumerate}
\end{prop}
\begin{proof}
\par States in $\textbf{P}^{\Gamma}$ where $s=1,2,\dots,\Gamma-1$ correspond to the states in the truncated Markov Chain $\textbf{A}^{\Gamma}$ where the source selected as the pivot has the same age. The system visiting these corresponding states in $\textbf{P}^{\Gamma}$ and $\textbf{A}^{\Gamma}$ is merely the same event, therefore the steady state probabilities and the transition probabilities for these states are equal. Therefore, they follow  (\ref{eq:piM}).
\par Now, for the states in $\textbf{P}^{\Gamma}$ for which $s \geq \Gamma$, we will prove that steady state probabilities exist. In order to do this, we define a \textit{augmented truncated Markov Chain} $ \{\textbf{A}^{s,\Gamma}[t], t\geq 1\}$, in which the only difference with the pivoted Markov Chain is that now the pivot source is truncated at $s+1$. At this point we consider the state  where the state of the pivot source is $s+1$ and the state of all the other sources are $\Gamma$ in the augmented truncated Markov Chain $ \{\textbf{A}^{s,\Gamma}[t], t\geq 1\}$. Then we realize that this specified state can be reached by any other state in the augmented truncated Markov Chain including itself, given that none of the last $s$ consecutive time slots resulted in a successful transmission. This is an event with non-zero probability. Thus, there is a single recurrent class and a unique steady state distribution for the augmented truncated Markov Chain. Finally, since the states in the augmented truncated Markov Chain have one-to-one correspondence with the states in the pivoted Markov Chain, the existence of a unique steady state distribution for the augmented truncated Markov Chain proves the existence of a unique steady state distribution for the states in the pivoted Markov Chain. 
\end{proof}
\begin{deftn}
Let the type of a state in $\textbf{P}^{\Gamma}$ be defined as $\textrm{T}^{\textbf{\textrm{P}}}\langle S^{\textbf{P}}\rangle = (s,m,\{u_{1}, u_{2}, \ldots, u_{n-m-1} \})$ where the $\{u_i\}$ are ordered from largest to smallest. Then $Q(S^{\textbf{P}})$, \textit{preceding} type of $S^{\textbf{P}}$, is defined as $\textrm{T}^{\textbf{\textrm{P}}}\langle S^{\textbf{P}}\rangle$ if $s=1$, as $(s-1,m,\{\Gamma-1,u_{1}-1, u_{2}-1,\ldots,u_{n-m-2}-1\})$ if $s \neq 1, u_{n-m-1} = 1$ and as $(s-1,m,\{u_{1}-1, u_{2}-1,\ldots,u_{n-m-1}-1\})$ if $\neq 1, u_{n-m-1} \neq 1$.
\end{deftn}
As can be seen from its definition, $Q(S^{\textbf{P}})$ is defined as the preceding type of $S^{\textbf{P}}$ given that the number of active sources (excluding the pivot source), $m$, does not change. This reasoning does not hold for the case $s=1$, nonetheless, since this case is not particularly the point of interest, we choose $Q(S^{\textbf{P}})$ to be the same type with $S^{\textbf{P}}$. Now that we have covered all possibilities for $Q(S^{\textbf{P}})$, we finally note that we will use $\pi(S^{\textbf{P}})$ or $\pi (s,m,\{u_{1}, u_{2}, \ldots, u_{n-m-1} \})$ to represent the steady state probability of $S^{\textbf{P}}$.
\begin{lemma} \label{lemma:3}
  Choose two arbitrary states in $\textbf{P}^{\Gamma}$, $S_1^{\textbf{P}}$ and $S_2^{\textbf{P}}$, where the state of the pivot source is equal for both states. Let the types of $S_1^{\textbf{P}}$ and $S_2^{\textbf{P}}$ be:
$$
    \textrm{T}^{\textbf{\textrm{P}}}\langle S_1^{\textbf{P}}\rangle = (s,m_1,\{u_{1}, u_{2}, \ldots, u_{n-m_1-1} \})
$$
$$
    \textrm{T}^{\textbf{\textrm{P}}}\langle S_2^{\textbf{P}}\rangle = (s,m_2,\{v_{1}, v_{2}, \ldots, v_{n-m_2-1} \})
$$
\begin{enumerate} [i)]
    \item Let $Q_1^{\textbf{P}}$ be any state satisfying $\textrm{T}^{\textbf{\textrm{P}}}\langle Q_1^{\textbf{P}}\rangle =  Q(S_1^{\textbf{P}})$. Then, 
    \begin{equation}
        \lim_{n \to \infty} \frac{\pi(S_1^{\textbf{P}})}{\pi(Q_1^{\textbf{P}})} = 1
    \end{equation}
    \item If $m_1=m_2$, then
    \begin{equation}
        \lim_{n \to \infty} \frac{\pi(S_1^{\textbf{P}})}{\pi(S_2^{\textbf{P}})} = 1
    \end{equation}
        \item If $m_1=m_2+1$, then
    \begin{equation}
        \lim_{n \to \infty} \frac{\pi(S_1^{\textbf{P}})}{n\,\pi(S_2^{\textbf{P}})} = \frac{1}{\alpha{e^{-k\alpha}}+\alpha\tau_2(e^{-\tau_2k\alpha}-e^{-k\alpha})} - k
    \end{equation}
\end{enumerate}
where $\lim_{n \to \infty} \frac{m_1}{n} = k$ and $\lim_{n \to \infty} \tau_1 n = \alpha$. ($k,\alpha \in \mathbb{R}^+$)
\end{lemma}
\begin{proof}
See Appendix \ref{app:lemma-f(k)}.
\end{proof}
\begin{theorem} \label{lemma:f(k)}
For some $r,\alpha \in \mathbb{R}^+$, such that $\lim_{n \to \infty} \frac{\Gamma}{n} = r$ and $\lim_{n \to \infty} \tau_1 n = \alpha$, define $f:(0,1)\to \mathbb{R}$:
\begin{equation}
\begin{aligned}
f(x) &= \ln(\frac{1}{{x\alpha}e^{-x\alpha}+{x\alpha}\tau_2(e^{-\tau_2x\alpha}-e^{-x\alpha})} - 1) \\&+ \ln(\frac{r}{x+r-1}-1)
\end{aligned}
\end{equation}
Then, for all $m$ such that $\lim_{n \to \infty} \frac{m}{n} = k \in (0,1)$ and $ s \in \mathbb{Z}^+$ 
\begin{equation}
    \lim_{n \to \infty} \ln \frac{P_m^{(s)}}{P_{m-1}^{(s)}} = f(k)
\end{equation}
where $P_m^{(s)}$ is the steady state probability of having $m$ active sources (excluding the pivot source), where state of the pivot source is $s$.
\end{theorem}
\begin{proof}
The total steady state probability of the states with $m$ active sources where the pivot source is in the state $s$ is $P_m^{(s)}$. The total number of such states is given as: 
\begin{equation}
    N_m = \binom{n-1}{m}\frac{(\Gamma-1)!}{(\Gamma-n+m)!} 
\end{equation}
Likewise, the total number of states with $m-1$ active sources where the pivot source is in the state $s$ is:
\begin{equation}
    N_{m-1} = \binom{n-1}{m-1}\frac{(\Gamma-1)!}{(\Gamma-n+m-1)!} 
\end{equation}
Then, the following gives the desired result
\begin{equation} \small
\hspace{-0.3cm}
\begin{aligned}
    &\lim_{n \to \infty} \frac{P_m^{(s)}}{P_{m-1}^{(s)}} = \lim_{n \to \infty} \frac{\sum\limits_{i=1}^{N_m} \pi(S_i^{(m)})}{\sum\limits_{j=1}^{N_{m-1}} \pi(S_j^{(m-1)})} \\
    &\stackrel{(a)}{=} \lim_{n \to \infty} \frac{n \sum\limits_{i=1}^{N_m} \left[\pi(S_i^{(m)}) / n\pi(S_1^{(m-1)})\right]}{\sum\limits_{j=1}^{N_{m-1}} \left[\pi(S_j^{(m-1)}) / \pi(S_1^{(m-1)})\right]} \\
    &\stackrel{(b)}{=} \lim_{n \to \infty} \frac{n \sum\limits_{i=1}^{N_m} (\frac{1}{\alpha{e^{-k\alpha}}+\alpha\tau_2(e^{-\tau_2k\alpha}-e^{-k\alpha})} - k)}{\sum\limits_{j=1}^{N_{m-1}} 1} \\
    &= \lim_{n \to \infty} \frac{n N_m(\frac{1}{\alpha{e^{-k\alpha}}+\alpha\tau_2(e^{-\tau_2k\alpha}-e^{-k\alpha})} - k)}{N_{m-1}} \\
    &= \lim_{n \to \infty} \frac{n (n-m)(\frac{1}{\alpha{e^{-k\alpha}}+\alpha\tau_2(e^{-\tau_2k\alpha}-e^{-k\alpha})} - k)}{m(\Gamma-n+m)} \\
    &= \left(\frac{1}{{k\alpha}e^{-k\alpha}+{k\alpha}\tau_2(e^{-\tau_2k\alpha}-e^{-k\alpha})} - 1\right) \left(\frac{1-k}{r+k-1}\right)
\end{aligned}
\end{equation}
\begin{equation}\small \notag
    \,
\end{equation}
where in the (a) step both sides of the fraction are divided to the steady state probability of a state with $m-1$ active sources where the pivot source is in the state $s$, and (b) follows from Lemma \ref{lemma:3} (ii) and (iii). Hence,
\begin{equation}
\begin{aligned}
    \lim_{n \to \infty} \ln \frac{P_m^{(s)}}{P_{m-1}^{(s)}} &= \ln (\frac{1}{{k\alpha}e^{-k\alpha}+{k\alpha}\tau_2(e^{-\tau_2k\alpha}-e^{-k\alpha})} - 1) \\&+ \ln (\frac{r}{r+k-1}-1) = f(k)
    \end{aligned}
\end{equation}
\end{proof}
What we essentially discovered is that as $n \to \infty$, the relation $P_m^{(s)} / P_{m-1}^{(s)}$ solely determines the distribution of $m$, no matter what the $s$ value is. This means that the number of active sources excluding the pivot source, $m$, is independent of the state of the pivot source. This result is formally expressed in the following corollary: 
\begin{corol}
\label{cor:pivot}
In the limit of a large network ($n \to \infty$), 
\begin{enumerate}[(i)]
    \item The number of active sources, $m$, (excluding the pivot) and the state of the pivot source, $s$, are independent.
    \item Given that the pivot source is active, ${\tau_1(1-\tau_1)^{m-1}}+{\tau_1\tau_2[(1-\tau_1\tau_2)^{m-1}-(1-\tau_1)^{m-1}]}$ is the probability of a successful transmission being made by the pivot source with no dependency on $s$.
    \item The probability of the pivot state being reset to 1 given that the pivot is active is $q_s = \lim_{l \to \infty} \sum\limits_{m=0}^l P_m^{(s)}{(\tau_1(1-\tau_1)^{m-1}}+{\tau_1\tau_2[(1-\tau_1\tau_2)^{m-1}-(1-\tau_1)^{m-1}])}$.
    
\end{enumerate}
\end{corol}
\begin{proof}
Parts (i) and (ii) follow from the proof of Lemma \ref{lemma:f(k)}. Since the distribution of the number of active sources $m$ and the state of the active source $s$ are independent, no mater what the $s$ value is, the pivot source observes the same number of active sources. Hence,  the  transition probabilities from $s=i$ to $s=i+1$ for $i < \Gamma$, and the transition probability from $s \geq \Gamma$ to 1 depends only on the number of active users. This means the evolution of the pivot source state $s$ is like shown in Fig. \ref{fig:markov}. 
\end{proof}
\begin{figure}[!htbp] 
\vspace{-0.7cm}
\centering
\includegraphics{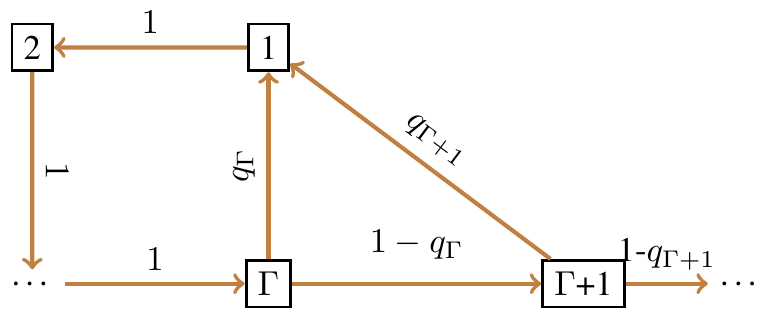}
\caption{State transition diagram of the pivot source}
\label{fig:markov}
\end{figure}
\par The transitions to the state 1 in Fig. \ref{fig:markov} represent successful transmissions made by the pivot source. In the rest, we will consider the asymptotic case of a large network as $n$ grows. All the transition probabilities to 1 (i.e. all the successful transitions of the pivot source) will have the same probability $q_o$ when $n\to \infty$, which will be showed later.
\subsection{Behaviour of MiSTA in Large Network Limit} \label{sect:rootAnalysis}
\par In this subsection, we compute the PMF of the number of active sources, $m$, at steady state under MiSTA, as $n\to \infty$. We begin by investigating the properties of the function $f$ defined in Theorem \ref{lemma:f(k)}, as it will provide valuable insight on the distribution of $m$. With this methodology, we prove that the fraction of active sources, $k$, converges in probability to one of the roots of $f$. 
\par In the asymptotic analysis, we re-express some of the parameters of the model, $\tau_1$ and $\Gamma$ to control their scaling  with $n$.
\begin{equation} \label{eq:30}
\alpha = n\tau_1, \;\; r = \Gamma/n,  \;\;  k = m/n
\end{equation}
\par As a result of this scaling model, the expected number of sources that will make a transmission attempt with probability $\tau_2$ in the data slot converges to a real number. Hence, we do not scale the value $\tau_2$ according to $n$. In the resulting system model; $k$, the fraction of users that are active, is the free variable while $\alpha$, $r$ and $\tau_2$ are fixed parameters that control the policy.  Due to the way it has been defined, $k$ will take values between 0 and 1, and it is an indicator of the instantaneous system load. The use of the $f(k)$ function at this point is due to the fact that roots of this function correspond to the local extrema of $P_m$. Specifically, decreasing roots of $f(k)$ correspond to the local maxima of $P_m$, where both $\ln {P_m}/{P_{m-1}}$ and $\ln {P_m}/{P_{m+1}}$ are positive.
\par Then, the number of roots $f(k)$ has is what restricts the number of local maxima for $P_m$. When the system was run using combinations of system parameters, it was observed that $f(k)$ may have more than 3 roots. Nevertheless, for more than 3 roots the AoI performance of the system is significantly poor compared to 1 or 3 roots. Thus, in our analysis we consider these two cases with desirable AoI results, which correspond to one local maximum and two local maxima for $P_m$, respectively. Despite their similarity, we can analyze these cases apart starting with one local maximum. The analysis is identical to the one that leads to the following result in \cite{Orhan}: 
\begin{theorem} \label{thm:1} \cite{Orhan}
	Let $m$ be the number of active sources and $k_0$ be the only root of $f(k)$. For the sequence $\epsilon_n = c n^{-1/3}$ where $c \in \mathbb{R}^+$,
	\begin{equation} \label{eq:thm1}
	\Pr(|\frac{m}{n} - k_0| < \epsilon_n) \to 1
	\end{equation}
\end{theorem}

\par Thm. \ref{thm:1} ensures that MiSTA asymptotically converges to a thinned slotted ALOHA policy, where only $nk$ sources are active at each slot. Since the fraction of active sources $k$ converges to $k_0$, this thinned slotted ALOHA scheme has approximately $nk_0$ active sources at each slot, independently of the individual states of the sources. This reduction in the number of contenders for a slot increases the success probability, thus improving throughput and consequently lowering AoI. 
\par The two local maxima case can similarly be analyzed by using Theorem \ref{thm:2} from \cite{Orhan}, which is based on Theorem \ref{thm:1}, except containing an additional requirement on the integral of the function $f$.
\par The difference between one local maximum and two local maxima cases is that in the two local maxima case, there are two decreasing roots of $f(k)$ that $k$ may converge to, which are the smallest and largest roots. As presented in Theorem \ref{thm:2}, the one that $k$ actually converges in probability is determined by the sign of the integral of $f(k)$ between these two roots.
\begin{theorem} \label{thm:2}\cite{Orhan}
	Let $k_0,k_1,k_2$ be all three and distinct roots of f(k) in increasing order and $m$ be the number of active sources.  
	\begin{enumerate}[i)]
	    \item If the integral of $f$ taken from $k_0$ to $k_2$ is negative, then for the sequence $\epsilon_n = c n^{-1/3}$ where $c \in \mathbb{R}^+$,
	\begin{equation}
	\Pr(|\frac{m}{n} - k_0| < \epsilon_n) \to 1
	\end{equation}
	\item If the integral of $f$ taken from $k_0$ to $k_2$ is positive, then for the sequence $\epsilon_n = c n^{-1/3}$ where $c \in \mathbb{R}^+$,
	\begin{equation}
	\Pr(|\frac{m}{n} - k_2| < \epsilon_n) \to 1
	\end{equation}
	\end{enumerate}
\end{theorem}
\par Note that the most favorable case is for $k$ to converge to $k_0$, so that the number of active sources is the smallest possible. Then, the system parameters should be selected accordingly.
\subsection{Optimal Average AoI for Large Networks}
\label{sect:ageCalc}
\begin{theorem} \label{thm:ageResult}
    In the large network limit (i.e. $n\to \infty$), optimal parameters of MiSTA satisfy the following:
\begin{equation}
\lim_{n \to \infty} \frac{\Gamma^*}{n} = 1.59
\end{equation}
\begin{equation}
\lim_{n \to \infty} n\tau_1^* = 10
\end{equation}
\begin{equation}
\tau_2^* = 0.38
\end{equation}
Furthermore, the optimal expected AoI at steady state scales with $n$ as:
\begin{equation}
\lim_{n \to \infty} \frac{\Delta^*}{n} = 0.9641
\end{equation}
\end{theorem}
\begin{proof}
At the ending of section \ref{sect:PivtChn}, $q_0$ was defined as the successful transmission probability of an active source, moreover, it has been argued that this value is fixed, or independent of age, in the steady state as the number of active sources converge to a value. We can also express $q_0$ in the following way:
\begin{equation}
    q_0 = \mathbb{E} [{\tau_1(1-\tau_1)^{M-1}}+{\tau_1\tau_2[(1-\tau_1\tau_2)^{M-1}-(1-\tau_1)^{M-1}]}]
\end{equation}
where the expectation is over the PMF of the number of active sources, $M$, at steady state, which was analyzed earlier.
We will first prove that 
\begin{equation}
\lim_{n \to \infty} n\,q_0 = \alpha e^{-k_0 \alpha}+ \alpha\tau_2(e^{-\tau_2k_0 \alpha}-e^{-k_0 \alpha})
\end{equation}
Let $\gamma_n$ be defined as: 
\begin{equation} \label{eq:44}
\gamma_n \triangleq \Pr(m_0-cn^{2/3} < M < m_0+cn^{2/3})
\end{equation} 
where $m_0 = k_0 n$. From Theorems \ref{thm:1} and \ref{thm:2}, $\gamma_n \to 1$ as $n \to \infty$. When $M$ satisfies the bounds  given in (\ref{eq:44}), the successful transmission probability, $q_0$, is also bounded. With this information, the following bound on $q_0$ is found: 
\begin{equation} \small
\begin{aligned}
 \gamma_n[(\tau_1[&(1-\tau_1)^{m_0}(1-\tau_2)+\tau_2(1-\tau_1\tau_2)^{m_0}])(1-\tau_1)^{-cn^{2/3}}] < q_0\\ < \gamma_n[&(\tau_1[(1-\tau_1)^{m_0}(1-\tau_2)+\tau_2(1-\tau_1\tau_2)^{m_0}])(1-\tau_1)^{cn^{2/3}}]\\&+(1-\gamma_n)   
\end{aligned}
\end{equation}
As $n\to \infty$, both upper and lower bounds converge to
${\tau_1(1-\tau_1)^{m_0}}+{\tau_1\tau_2[(1-\tau_1\tau_2)^{m_0}-(1-\tau_1)^{m_0}]}$. Finally,
\vspace{0.3cm}
\begin{equation} \small
\begin{aligned}
\lim_{n \to \infty} n\,q_0 &= \lim_{n \to \infty} n\tau_1[(1-\tau_1)^{m_0}(1-\tau_2)+\tau_2(1-\tau_1\tau_2)^{m_0}] \\&= \alpha e^{-k_0 \alpha}+ \alpha\tau_2(e^{-\tau_2k_0 \alpha}-e^{-k_0 \alpha})
\end{aligned}
\end{equation}
Since the transitions between the states of a source is as given in Fig. \ref{fig:markov}, value of $q_0$ can be used to compute the steady state probabilities of the states. Furthermore, the states correspond one-to-one with the ages the source have. Then, finding the steady state probabilities of the states is merely finding the steady state probabilities of the ages. With this methodology, the steady state probability of state $j$ is:
\begin{equation}
\pi_j = \frac{(1-q_0)^{max\{j-\Gamma,0\}}}{\Gamma-1+1/q_0}, \hspace{5mm} j = 1,2,\ldots
\end{equation}
The expected time-average AoI expression is found using the steady state probabilities of the ages: 
\begin{equation}
\Delta = \frac{\Gamma(\Gamma-1)}{2(\Gamma-1+1/q_0)} + 1/q_0
\end{equation}
The average AoI is also expressed in the limit of large network as:
\begin{equation} \label{eq:ageformula}
\begin{aligned}
\lim_{n \to \infty} \frac{\Delta}{n} &= \frac{r^2}{2(r+\frac{1}{\alpha e^{-k_0 \alpha}+ \alpha\tau_2(e^{-\tau_2k_0 \alpha}-e^{-k_0 \alpha})})} \\&+ \frac{1}{\alpha e^{-k_0 \alpha}+ \alpha\tau_2(e^{-\tau_2k_0 \alpha}-e^{-k_0 \alpha})}
\end{aligned}
\end{equation}
The system parameters $r$ and $k_0$ can be used to re-express(\ref{eq:ageformula}) as:
\begin{equation} \label{eq:alternativeAge}
\lim_{n \to \infty} \frac{\Delta}{n} = r \frac{k_0^2+1}{2(1-k_0)}
\end{equation}
With the right selection of system parameters, the Average AoI expression can be minimized. 
\end{proof}
\par Analyzing (\ref{eq:ageformula}), optimal parameters and some other steady-state characteristics such as $k_0$ and average AoI, are derived for MiSTA. These findings are summarized in Table I with corresponding values for Threshold ALOHA and slotted ALOHA for comparison. Since both Threshold ALOHA and MiSTA have two regimes of operation, namely two local maxima case and single local maximum case; the results for these regimes are provided seperately. 
\par As presented in Table I, the throughput in Slotted ALOHA cannot exceed $e^{-1}$ and minimum average AoI scales with $n$ as $en$. The average AoI value drops to nearly half this value in Threshold ALOHA policy, while approximately maintaining the throughput level of slotted ALOHA. MiSTA, on the other hand, significantly increases the throughput value and decreases the average AoI to an even smaller value.
\vspace{-0.5 cm}
\begin{center}

\begin{displaymath}
\begin{small}
\begin{array}{|l|c|c|c|c|c|c|} \hline 
\, & r^* & \alpha^* & \tau_2^* & k_0^* & \Delta^*/n & Thr. \\ \hline 
\textrm{MiSTA(SP)} & 1.59  & 9.8 & 0.37 & 0.1565 & 0.9656 & 0.5252 \\ \hline
\textrm{MiSTA(DP)} & 1.59  & 10 & 0.38 & 0.1555 & \textbf{0.9641} & 0.5266 \\ \hline
\textrm{TA(SP)} & 2.17  & 4.43 & - & 0.2052 &  1.4226 & 0.3658 \\ \hline
\textrm{TA(DP)} & 2.21  & 4.69 & - & 0.1915  & \textbf{1.4169} & 0.3644 \\ \hline 
\textrm{SA}    & 0  & 1 & - & 1 & e  & e^{-1}\\ \hline \end{array}
\end{small}
\end{displaymath}
\vspace{0 cm}
\captionof{table}{A comparison of optimized parameters of ordinary slotted ALOHA, Threshold ALOHA and mini slotted Threshold ALOHA and the resulting AoI and throughput values. $r^*$: age-threshold$/n$; $\tau_2^*$: probability of transmission in the second toss; $\alpha^*$: transmission probability$\times n$; $k_0^*$: expected fraction of active users; $\Delta^*$: avg. AoI}
\end{center}
\subsection{Derivation of a Lower Bound for AoI} \label{sect:lowerbound}  
\par In this subsection, we derive the maximum throughput attainable by the MiSTA policy and use the resulting value to obtain a lower bound on the minimum time average AoI achievable. As throughput corresponds to the long term average rate of transmissions in the network, a lower bound on average AoI in the network was shown to be bounded in terms of $q^{max}$, the maximum achievable throughput, in \cite{shirin}:
\begin{equation} \label{eq:sh-limit}
	    \frac{\Delta}{n} \geq \frac{1}{2q^{max}}+\frac{1}{2n}
\end{equation}
The idea in MiSTA is to lower AoI below that attained by Threshold ALOHA, by increasing the throughput with respect to the latter. As stated earlier, $\tau_1$  and $\tau_2$ are the attempt probabilities of the sources that are active in the mini slot and the data slot, respectively. At a given instant, let the number of active sources be $m$ and let $q$ be the instantaneous throughput of MiSTA. Then $q$ can be written as:
\begin{equation}
    q = {m\tau_1(1-\tau_2)(1-\tau_1)^{m-1}}+m\tau_1\tau_2(1-\tau_1\tau_2)^{m-1}
\end{equation}
Our analysis will mainly focus on the asymptotic case where $n$ and $m$ values are large. We define $G \triangleq m\tau_1$ and rewrite $q$ in the limit of large $m$ as:
\begin{equation}
    q = {\tau_2}Ge^{-{\tau_2}G}+{(1-\tau_2)Ge^{-G}}
\end{equation}
The above expression is optimized with the following choice of parameters:
\begin{equation} \label{eq:thr-optimal}
    q^{max} = 0.5315,\quad G^*=1.59,\quad \tau_2^*=0.38
\end{equation}
Finally, we use \eqref{eq:sh-limit} and \eqref{eq:thr-optimal} to obtain the following proposition.
\begin{prop} \label{prop:low-limit}
Average AoI in MiSTA is lower bounded by $0.9407n + 0.5$.
\end{prop}
The increase in the maximum achievable throughput is due to the novel time structure detailed in the beginning of Sec. \ref{sect:minislot}. Since the maximum achievable throughput in slotted ALOHA and TA is $e^{-1}$, the age lower bound for these policies is given as $1.3591n + 0.5$. Note that the increase in maximum achievable throughput for MiSTA results in an age lower bound which is not obtainable by these policies. Further, we observe that MiSTA policy is an effective way of utilizing the mini slot as the difference between the lower limit of Prop. \ref{prop:low-limit} and the optimal AoI of Table I is less than 2.5\%. The loss of throughput under age optimal MiSTA is less than 1\%.

\subsection{Spectral Efficiency Analysis of MiSTA} \label{sect:spectral}
\par When we propose prepending a mini slot to each data slot, one of the first concerns is to conserve the spectral efficiency of the  system. In this section we show that perhaps contrary to the immediate intuition MiSTA increases the spectral efficiency of the system, especially for large data slots. The spectral efficiency lost in mini slots is compensated by the increase in the throughput in almost all cases of practical interest. We will present these results using the notation in Table II.
\vspace{-0.5 cm}
\begin{center}
\begin{small}
\begin{displaymath}
\begin{array}{|l|c|} \hline 
 \eta & \textrm{Spectral efficiency of the Threshold ALOHA(bits/s/Hz)}  \\ \hline
 \eta' & \begin{aligned}\textrm{Spectral efficiency of } &\textrm{the MiSTA (bits/s/Hz)} \end{aligned}\\ \hline
B & \textrm{Channel Bandwidth}  \\ \hline
H & \textrm{Time Horizon}   \\ \hline
T_b & \textrm{The time it takes to send 1 bit (s)}  \\ \hline
\theta_1 & \textrm{Throughput of Threshold ALOHA}  \\ \hline
\theta_2 & \begin{aligned} \textrm{Throughput of mini slotted Threshold ALOHA} \end{aligned} \\ \hline
c & \textrm{Number of bits in the data slot}  \\ \hline
d & \textrm{Number of bits in the mini slot}   \\ \hline
\end{array}
\end{displaymath}
\vspace{0 cm}
\end{small}
\captionof{table}{Notation used in the spectral efficiency analysis}
\end{center}
\par With the definitions given Table II, we first derive the following expressions.
\begin{equation}  
\eta = \frac{H{\theta_1}c}{HBc{T_b}} = \frac{\theta_1}{B{T_b}}
\end{equation}
\begin{equation}  
\eta' = \frac{H{\theta_2}c}{HB(c+d){T_b}} = \frac{{\theta_2}c}{(c+d)B{T_b}}
\end{equation}
\begin{equation}  
\frac{\eta'}{\eta} = \frac{\theta_2}{\theta_1}\frac{c}{c+d}
\end{equation}
\par In order to preserve the spectral efficiency, the $\eta/\eta'$ expression should at least be equal to 1. Some representative values that the ratio $\theta_2/\theta_1$ takes are provided in Table II.
\vspace{-0.5 cm}
\begin{center}
\begin{displaymath}
\begin{array}{|l|c|c|c|} \hline 
\, & \textrm{MiSTA}(\theta_2) & \textrm{TA}(\theta_1) & \theta_2/\theta_1 \\ \hline 
\textrm{1000 Sources} & 0.5251 & 0.3632 & 1.448  \\ \hline
\textrm{500 Sources} & 0.5179 & 0.3581 & 1.446 \\ \hline
\textrm{100 Sources} & 0.5019 & 0.3633 & 1.382  \\ \hline \end{array}
\end{displaymath}
\vspace{0 cm}
\captionof{table}{Typical values the $\theta_2/\theta_1$ expression takes for three different $n$ values, for MiSTA and Threshold ALOHA.}
\end{center}
\par Consistently with the representative values of $\theta_2/\theta_1$ given in Table III, it can be shown that as long as the ratio $c/d$ is greater than $2.23$, MiSTA incurs no loss in spectral efficiency. 
\par In protocols which are currently used in real time systems such as IEEE 802.11, the value $c$ is typically a few Kbytes. Moreover, for the identification purposes in the mini slot, a length such as 128 bits is adequate for $d$. Hence, even for short packets around 2 Kbytes, with a moderate 128 bit ID header, the ratio $c/d$ is well above $2.23$.
\section{An Extension of MiSTA with Multiple Mini Slots: MuMiSTA}
We can define a natural extension of MiSTA, the \textit{Multiple Mini Slotted Threshold ALOHA (MuMiSTA)}, where a multiple number of mini slots are prepended to each slot. Such a policy achieves a throughput value of 95\% with just 32 mini slots. Subsequently, the average AoI scales with $n$ as $0.531n$. The MATLAB simulation results for MuMiSTA policy are presented in Fig. \ref{fig:mumista} together with the simulation results of \cite{kadotaInfo2021} for comparison. The results for \cite{kadotaInfo2021} are not smooth for large packet sizes since the system does not reach steady state in the time horizon of the simulations. In these simulations MuMiSTA is run for $10^7$ with 32 mini slots and 100 users. It is apparent that this protocol becomes even more efficient when data slots become noticeably longer than mini slots. Due to lack of space, a detailed study of MuMiSTA is omitted and left for future work.
 \begin{figure}[ht] 
\centering
\includegraphics[scale=0.55]{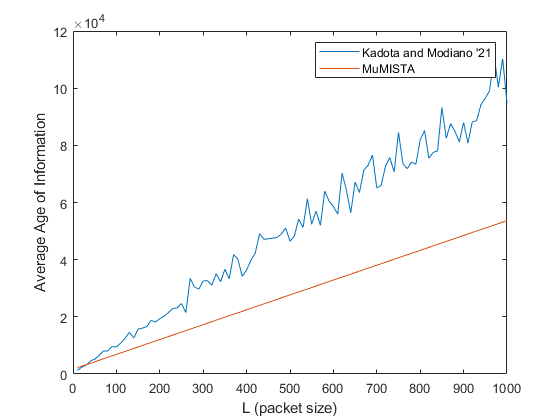}
\caption{Comparison of MuMiSTA and the reservation based random access policy in \cite{kadotaInfo2021}.}
\label{fig:mumista}
\end{figure}
\section{Numerical Results and Discussion} \label{sect:numerical}
\par In order to obtain analytical results, we modeled MiSTA with a Markov Chain whose states are vectors constructed with the individual ages of the sources. The evolution of this vector with time was then analyzed for results. In this section, we obtain simulation results by building the same system model in MATLAB and observing the evolution of a vector consisting of ages of the individual sources. The same simulation environment is also used to get corresponding values for slotted ALOHA and Threshold ALOHA, in order to make performance comparisons. First, the simulation for MiSTA is run with 300 sources for $10^7$ time slots in order to obtain the pmf of $m$, the number of active sources. The system parameters are chosen as the optimal double peak case parameters in Table I. In Fig. \ref{fig:Lemma_1_Figure}, the result of this simulation is presented as a histogram together with the theoretical curve for the PMF, which is plotted using Lemma \ref{lemma:truncated}. We see that simulation results match the theoretical findings. 
\begin{figure}[ht] 
\centering
\includegraphics[scale=0.45]{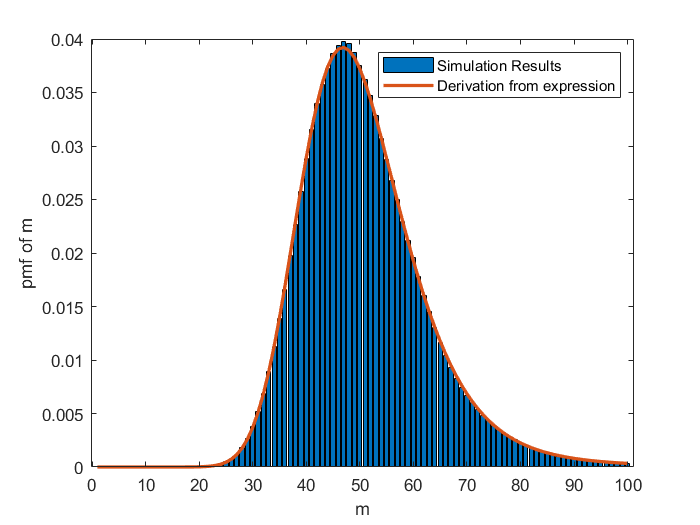}
\caption{The PMF of \textit{m}, instantaneous number of active sources, obtained with both simulations (histogram) and theoretical expression (\ref{eq:piM}) under 100 users.}
\label{fig:Lemma_1_Figure}
\end{figure}
\par Then, the evolution of the fraction of active sources, $k$, is observed with the same setting as time slots progress, and the results are plotted in Fig. \ref{fig:kvst}. We see that in less than $10^5$ time slots, the $k$ value converges to a specific value as theoretically proven in Theorems \ref{thm:1} and \ref{thm:2}. In addition, we see that when we multiply the number of sources with the value $k$ converges, we get the $m$ value where the pmf peaks in Fig. \ref{fig:Lemma_1_Figure} as expected.
\begin{figure}[ht] 
\centering
\includegraphics[scale=0.45]{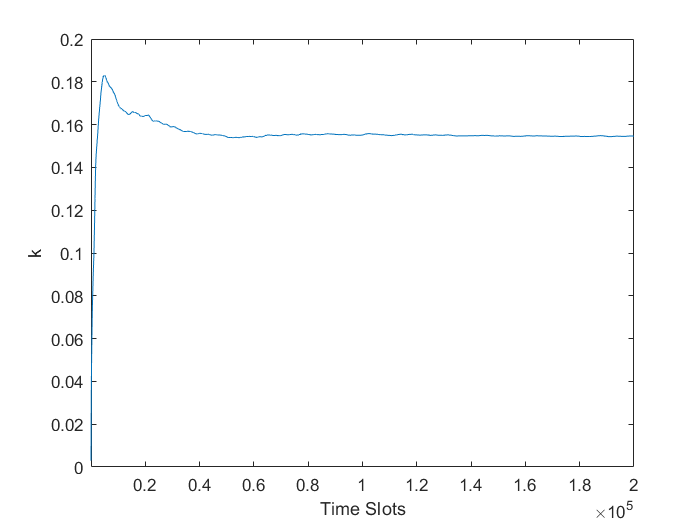}
\caption{The time evolution of instantaneous load, \textit{k}, as MiSTA progresses and converges.}
\label{fig:kvst}
\end{figure}
\par After the simulation with the double peak case optimal parameters, the parameters are changed to single peak optimal ones for comparison. The instantaneous AoI is plotted together for both sets of parameters as time slots progress, as given in Fig. \ref{fig:Converge_age}. We see that it takes around 250 thousand slots for the single peak optimal case to converge to the steady state AoI, whereas this value goes up to around 600 thousand slots for the double peak optimal case. Hence, while the double peak optimal case may be superior, it may take significantly longer time to converge. Thus, the single peak optimal case may be more suitable, especially when the network size is moderate or small. 
\begin{figure}[ht] 
\centering
\includegraphics[scale=0.45]{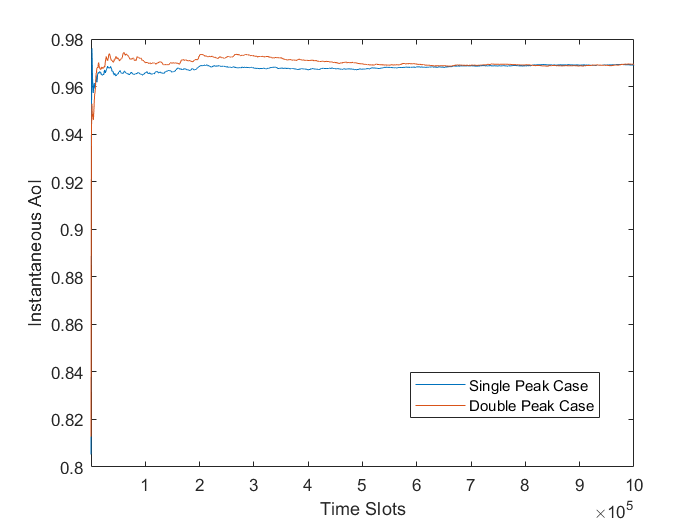}
\caption{The time evolution of average AoI as MiSTA operates in the single peak and double peak regimes, using optimal parameters in Table I.}
\label{fig:Converge_age}
\end{figure}
\par In Table I, we listed the optimal parameters for MiSTA, which are obtained analytically. Now, we run our simulation by sweeping through values of parameters $r$ and $\tau_2$ and by noting the resulting AoI values. Then, the average AoI is plotted against these parameters in Figs. \ref{fig:AoIvsr} and \ref{fig:AoIvstau2} respectively. The results in the figures are corresponding to the optimal parameters for double peak case in Table I, thus they verify our analytical findings.
\begin{figure} [ht]
\begin{minipage}{0.45\textwidth}
\includegraphics[width=\textwidth]{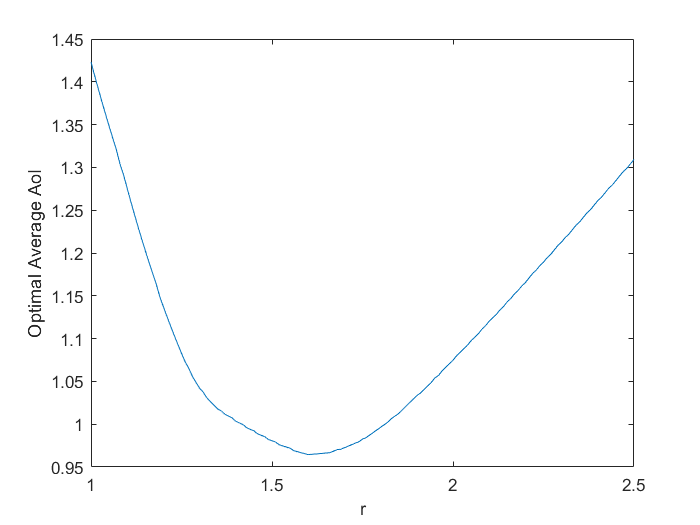}
\caption{Optimal average AoI vs $r$ plot.}
\label{fig:AoIvsr}
\end{minipage}%
\hspace{0.75cm}
\begin{minipage}{0.45\textwidth}
\includegraphics[width=\textwidth]{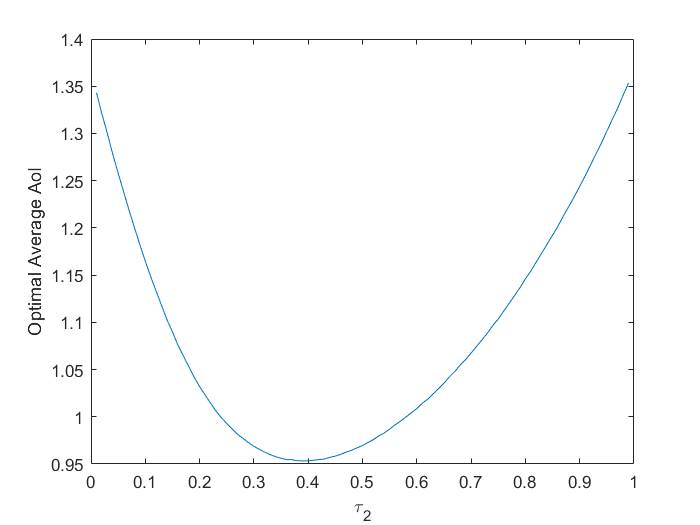}
\caption{Optimal average AoI vs $\tau_2$ plot.}
\label{fig:AoIvstau2}
\end{minipage}
\end{figure}
\par Finally, in the MATLAB simulations, we compare MiSTA with slotted ALOHA and Threshold ALOHA in terms of AoI and throughput performance. In these simulations, the number of sources ranges between 50 and 1000 and the system has been run for $10^7$ time slots with double peak case optimal parameters. First we compare the AoI performances of these policies in Fig. \ref{fig:AoIvsN}, where AoI is plotted against $n$. As expected, Threshold ALOHA reduces the minimum AoI achievable by slotted ALOHA to almost one half, while MiSTA outperforms both by reducing this value to roughly one third of achievable by slotted ALOHA. Then, we plot the throughput results for these policies in Fig. \ref{fig:throunum}, again plotting against $n$. As was the case for AoI performance, simulation results for throughput matchthe findings in Table I, where we had found that MiSTA increases the throughput to roughly $0.53$.  
\begin{figure}[ht] 
\centering
\includegraphics{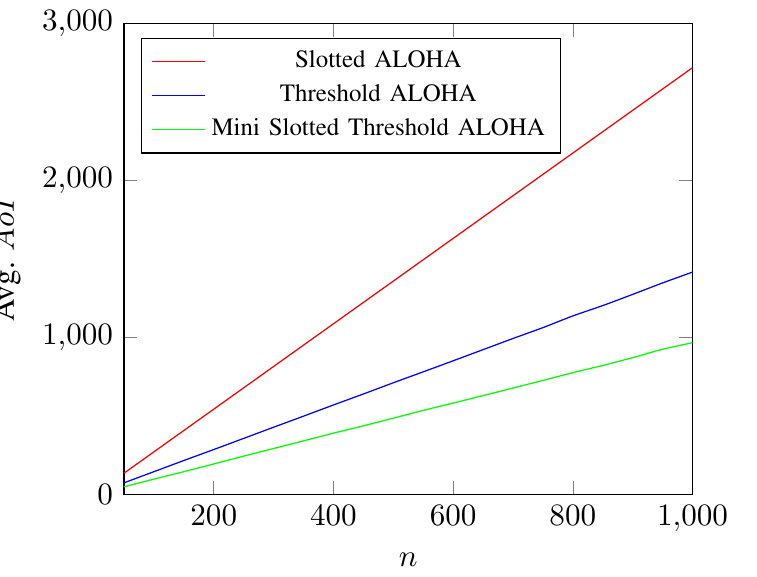}
\caption{Optimal time average $AoI$ vs $n$, number of sources, under Slotted ALOHA (computed), Threshold ALOHA (simulated) and mini slotted Threshold ALOHA (simulated).}
\label{fig:AoIvsN}
\end{figure}
\begin{figure}[ht] 
\centering
\includegraphics{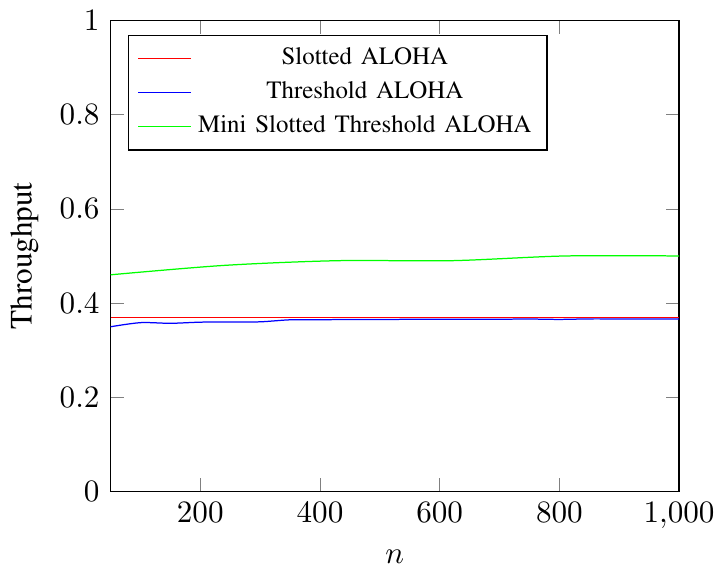}
\caption{Throughput vs $n$, number of sources, under Slotted ALOHA (computed), Threshold ALOHA (simulated) and mini slotted Threshold ALOHA (simulated).}
\label{fig:throunum}
\end{figure}
\vspace{-0.3cm}
\section{Conclusion} \label{sect:conclusion}
\par In this paper, a modification to the Threshold ALOHA policy, termed \textit{Mini Slotted Threshold ALOHA (MiSTA)}, was proposed. It was shown that, in a large network, MiSTA reduces the time average age of information attained in the network by $32\%$, while increasing throughput by $43\%$ over Threshold ALOHA. The policy achieves this through a change in slot structure, where users are given a chance to contend during a minislot placed in the beginning of any data slot, with success feedback given by the AP in the event of a successful transmission. The steady state age distribution of nodes was found. Using this, an average AoI expression was derived and optimized in terms of the policy parameters. The optimal parameters in the large network limit hav been summarized in Table I. The results show that MiSTA is equivalent to a thinned slotted ALOHA with around $15\%$ of all sources active at a slot and its optimal AoI value scales with the network size, $n$, as \textbf{0.9641n}. The minimum achievable AoI with MiSTA is around one third of that achievable by plain slotted ALOHA. In addition to its AoI performance, MiSTA increases the throughput value over slotted ALOHA by approximately  $45\%$ and increases the spectral efficiency of the system in all practical cases. It should be remembered that MiSTA uses success feedback per each successful transmission, which slotted ALOHA does not.
\par A possible direction for future work is to precisely formulate and analyze the \textit{MuMiSTA} policy briefly defined and numerically presented in this paper. Furthermore, analysis and optimization of MiSTA under various scenarios as stochastic arrivals, unreliable channels or contention resolution methods where multiple transmissions are permitted in a slot, are further directions that can be pursued.
\vspace{-0.3cm}


%

\appendices
\section{Proof of Lemma \ref{lemma:3}} \label{app:lemma-f(k)}
\par  We will begin the proof for the $s$ values $s=1,2,\ldots,\Gamma-1$. The properties $(i)$ and $(ii)$ directly follow from Prop. \ref{prop:pivot} (i), where  $\pi(S_1^{\textbf{P}}) = \pi_{m_1}$ and $\pi(S_2^{\textbf{P}}) = \pi_{m_2}$. Although Property $(iii)$ follows from the same property, it is not directly seen
\begin{equation} \label{eq:45}
\begin{aligned}
    &\lim_{n \to \infty} \frac{\pi(S_1^{\textbf{P}})}{n\,\pi(S_2^{\textbf{P}})} = \lim_{n \to \infty} \frac{\pi_{m_1}}{n\pi_{{m_1}-1}} \\
    &\stackrel{(a)}{=} \lim_{n \to \infty} \frac{1}{n\tau_1[(1-\tau_1)^{m_1-1}(1-\tau_2)+\tau_2(1-\tau_1\tau_2)^{m_1-1}]} \\&-\frac{(m_1-1)\tau_1[(1-\tau_1)^{m_1-2}(1+\tau_2)+\tau_2(1-\tau_1\tau_2)^{m_1-2}]}{n\tau_1[(1-\tau_1)^{m_1-1}(1-\tau_2)+\tau_2(1-\tau_1\tau_2)^{m_1-1}]} \\
    &= \frac{1}{\alpha{e^{-k\alpha}}+\alpha\tau_2(e^{-\tau_2k\alpha}-e^{-k\alpha})} - k
\end{aligned}
\end{equation} 
where the step (a) follows from (\ref{eq:piM}). 
\par Next, we move on to the $s$ value $s = \Gamma$ and show that the properties still hold. We will start with showing that $\pi(S_1^{\textbf{P}}) = \pi_{m_1}$. Assuming that $1 \not\in \{u_{1}, u_{2}, \ldots, u_{n-m_1-1} \}$ holds and $S_1^{\textbf{P}}$ is the current state, the previous state can be one of the following types: 
\begin{itemize}
    \item $(\Gamma-1,m_1,\{u_{1}-1, u_{2}-1, \ldots, u_{n-m_1-1}-1 \})$
    \item $(\Gamma-1,m_1-1,\{\Gamma-1,u_{1}-1, u_{2}-1, \ldots, u_{n-m_1-1}-1 \})$
\end{itemize}
 In Prop. \ref{prop:pivot} (ii), the steady state probabilities of these types are given as $\pi_{m_1}$ and $\pi_{m_1-1}$, respectively. The steady state probability of $S_1^{\textbf{P}}$ can be calculated using the given probabilities for the preceding state and their transition probabilities as:  
\begin{equation} \label{eq:36}
\begin{aligned}
\pi(S_1^{\textbf{P}}) &= \pi_{m_1}(1-m_1\tau_1(1-\tau_1)^{m_1-1}\\&-m_1\tau_1\tau_2[(1-\tau_1\tau_2)^{m_1-1}-(1-\tau_1)^{m_1-1}]) \\&+  \pi_{m_1-1}m_1(1-(m_1-1)\tau_1(1-\tau_1)^{m_1-2}\\&-(m_1-1)\tau_1\tau_2[(1-\tau_1\tau_2)^{m_1-2}-(1-\tau_1)^{m_1-2}])\\ &= \pi_{m_1}
\end{aligned}
\end{equation}
where the $\pi_{m_1}$ result is obtained through the ratio given in (\ref{eq:piM}). Now that we are done with the case of $1 \not\in \{u_{1}, u_{2}, \ldots, u_{n-m_1-1} \}$, we move on to the case $1 \in \{u_{1}, u_{2}, \ldots, u_{n-m_1-1} \}$. We will assume $u_{n-m_1-1} = 1$ without loss of generality and then follow similar steps with the preceding case. The previous state can be one of the following this time: 
\begin{itemize}
    \item $(\Gamma-1,m_1+1,\{u_{1}-1, u_{2}-1, \ldots, u_{n-m_1-2}-1 \})$
    \item $(\Gamma-1,m_1,\{\Gamma-1,u_{1}-1, u_{2}-1, \ldots, u_{n-m_1-2}-1 \})$
\end{itemize}
Again using Prop. \ref{prop:pivot} (ii), we find the steady state probabilities of these states as $\pi_{m_1+1}$ and $\pi_{m_1}$, respectively. Then, the steady state probability of $S_1^{\textbf{P}}$ is derived using the transition probabilities as: 
\begin{equation} \label{eq:37} 
\begin{aligned}
\pi(S_1^{\textbf{P}}&) \\=&\pi_{m_1+1}(\tau_1(1-\tau_1)^{m_1}+\tau_1\tau_2[(1-\tau_1\tau_2)^{m_1}-(1-\tau_1)^{m_1}]) \\&+  \pi_{m_1}m_1\tau_1[(1-\tau_1)^{m_1-1}(1+\tau_2)-\tau_2(1-\tau_1\tau_2)^{m_1-1}] \\=&\pi_{m_1}
\end{aligned}
\end{equation}
\par  $\pi(S_2^{\textbf{P}}) = \pi_{m_2}$ since there is a symmetry. The properties $(i)$ and $(ii)$ follow from Prop. \ref{prop:pivot} (i) and Property $(iii)$ follows from (\ref{eq:45}).
\par The last case of ranges for the $s$ value is $\forall s \geq \Gamma$. In order to show that the lemma still holds for this case, we will use induction. The preceding case $s = \Gamma$ is the initial case and is already covered. Therefore, we assume that $s>\Gamma$ and that the properties of the lemma holds for all the smaller values of $s$. Here we will again use two cases two make our analysis:\\
\textbf{Case 1}. If $1 \not\in \{u_{1}, u_{2}, \ldots, u_{n-m-1} \}$ \\
For the sake of readability, the probability expressions are shortened in the following way:
\begin{equation} \small
 \begin{aligned}
    \pi_m^{(s)}  &= \pi(s,m,\{u_{1}, u_{2}, \ldots, u_{n-m-1} \}) = \pi(S_1^{\textbf{P}}) \\
    \pi_m^{(s-1)} &= \pi(s-1,m,\{u_{1}-1, u_{2}-1, \ldots, u_{n-m-1}-1 \}) \\&= \pi(Q_1^{\textbf{P}}) \\
    \pi_{m-1}^{(s-1)} &= \pi(s-1,m-1,\{\Gamma-1,u_{1}-1, u_{2}-1, \ldots, u_{n-m-1}-1 \})
\end{aligned}   
\end{equation}
A state of type $(s,m,\{u_{1}, u_{2}, \ldots, u_{n-m-1} \})$ can be preceded by states with probabilities $\pi_m^{(s-1)}$ or $\pi_{m-1}^{(s-1)}$. Then, the value of $\pi_m^{(s)}$ can be calculated using the transition probabilities as:
\begin{equation}
    \begin{aligned}
\pi_m^{(s)} = &\pi_m^{(s-1)}(1-(m+1)\tau_1(1-\tau_1)^{m}\\&-(m+1)\tau_1\tau_2[(1-\tau_1\tau_2)^{m}-(1-\tau_1)^{m}]) \\+ &\pi_{m-1}^{(s-1)}(m+1)(1-m\tau_1(1-\tau_1)^{m-1}\\&-m\tau_1\tau_2[(1-\tau_1\tau_2)^{m-1}-(1-\tau_1)^{m-1}])
\end{aligned}
\end{equation}
Then, it is shown in (\ref{eq:lemma2long1}) that $\lim_{n \to \infty} \frac{\pi_m^{(s)}}{\pi_m^{(s-1)}}=1$ where $(a)$ follows from property $(iii)$.
\begin{equation} \label{eq:lemma2long1} \small
\vspace{0.5cm}
\begin{aligned}
    &\lim_{n \to \infty} \frac{\pi_m^{(s)}}{\pi_m^{(s-1)}} \\&= \lim_{n \to \infty} \frac{\pi_m^{(s-1)}(1-(m+1)\tau_1[(1-\tau_1)^m(1-\tau_2)+\tau_2(1-\tau_1\tau_2)^m]) }{\pi_m^{(s-1)}} \\
    &+\frac{\pi_{m-1}^{(s-1)}(m+1)(1-m\tau_1[(1-\tau_1)^{m-1}(1-\tau_2)+\tau_2(1-\tau_1\tau_2)^{m-1})}{\pi_m^{(s-1)}} \\
    &= \lim_{n \to \infty} 1-(m+1)\tau_1[(1-\tau_1)^m(1-\tau_2)+\tau_2(1-\tau_1\tau_2)^m]
    \\&+\frac{\pi_{m-1}^{(s-1)}}{\pi_m^{(s-1)}}(m+1)(1-m\tau_1[(1-\tau_1)^{m-1}(1-\tau_2)+\tau_2(1-\tau_1\tau_2)^{m-1}) \\
    &= \lim_{n \to \infty} 1-(\frac{m+1}{n}(n\tau_1))[(1-\tau_1)^m(1+\tau_2)-\tau_2(1-\tau_1\tau_2)^m] \\&+\frac{n\,\pi_{m-1}^{(s-1)}}{\pi_m^{(s-1)}}\frac{m+1}{n}(1-m\tau_1[(1-\tau_1)^{m-1}(1+\tau_2)-\tau_2(1-\tau_1\tau_2)^{m-1}]) \\
    &\stackrel{(a)}{=} \lim_{n \to \infty} 1-k\alpha e^{-k\alpha}-k\alpha\tau_2[e^{-\tau_2k\alpha}-e^{-k\alpha}] \\&+ \frac{1}{\frac{1}{\alpha{e^{-k\alpha}}+\alpha\tau_2(e^{-\tau_2k\alpha}-e^{-k\alpha})} - k}k(1-k\alpha e^{-k\alpha}[1+\tau_2e^{-\tau_2}-\tau_2]) = 1
\end{aligned}
\end{equation}
\\
\textbf{Case 2}. If $1 \in \{u_{1}, u_{2}, \ldots, u_{n-m-1} \}$, and without loss of generality we choose $u_{n-m-1}=1$ \\
Again for the sake of readability, the probability expressions are shortened in the following way:
\begin{equation}
  \begin{aligned}
    \pi_m^{(s)} &= \pi(s,m,\{u_{1}, u_{2}, \ldots, u_{n-m-2},1 \}) = \pi(S_1^{\textbf{P}}) \\
    \pi_m^{(s-1)} &= \pi(s-1,m,\{\Gamma-1,u_{1}-1, u_{2}-1, \ldots, u_{n-m-2}-1 \}) \\&= \pi(Q_1^{\textbf{P}}) \\
    \pi_{m+1}^{(s-1)} &= \pi(s-1,m+1,\{u_{1}-1, u_{2}-1, \ldots, u_{n-m-2}-1 \})
\end{aligned}  
\end{equation}
A state of type $(s,m,\{u_{1}, u_{2}, \ldots, u_{n-m-2},1 \})$ can be preceded by two types of states with steady state probabilities $\pi_m^{(s-1)}$ and $\pi_{m+1}^{(s-1)}$. Then, using the transition probabilities, the value of $\pi_m^{(s)}$ is obtained as:
\begin{equation}
\begin{aligned}
 \pi_m^{(s)} = &\pi_{m+1}^{(s-1)}(\tau_1(1-\tau_1)^{m}+\tau_1\tau_2[(1-\tau_1\tau_2)^{m}-(1-\tau_1)^{m}])\\ + &\pi_m^{(s-1)}(m\tau_1(1-\tau_1)^{m-1}-m\tau_1\tau_2[(1-\tau_1\tau_2)^{m-1}\\&-(1-\tau_1)^{m-1}])   
\end{aligned}
\end{equation}
Then, it is shown in (\ref{eq:lemma2long2}) that $\lim_{n \to \infty} \frac{\pi_m^{(s)}}{\pi_m^{(s-1)}}=1$.
\begin{equation} \label{eq:lemma2long2} \small
\begin{aligned}
    &\lim_{n \to \infty} \frac{\pi_m^{(s)}}{\pi_m^{(s-1)}} \\&= \lim_{n \to \infty} \frac{\pi_{m+1}^{(s-1)}\tau_1[(1-\tau_1)^{m}(1-\tau_2)+\tau_2(1-\tau_1\tau_2)^{m}]}{\pi_m^{(s-1)}} \\&+
    \frac{\pi_m^{(s-1)}(m\tau_1[(1-\tau_1)^{m-1}(1+\tau_2)-\tau_2(1-\tau_1\tau_2)^{m-1}])}{\pi_m^{(s-1)}} \\
    &= \lim_{n \to \infty} \frac{\pi_{m+1}^{(s-1)}}{\pi_m^{(s-1)}}(\tau_1(1-\tau_1)^{m}+\tau_1\tau_2[(1-\tau_1\tau_2)^{m}-(1-\tau_1)^{m}])\\&+(m\tau_1(1-\tau_1)^{m-1}-m\tau_1\tau_2[(1-\tau_1\tau_2)^{m-1}-(1-\tau_1)^{m-1}]) \\
    &= \lim_{n \to \infty}
    \frac{\pi_{m+1}^{(s-1)}}{n\pi_m^{(s-1)}}n\tau_1[(1-\tau_1)^{m}(1-\tau_2)+\tau_2(1-\tau_1\tau_2)^{m}]\\&+ (m\tau_1[(1-\tau_1)^{m-1}(1+\tau_2)-\tau_2(1-\tau_1\tau_2)^{m-1}]) \\
    &= \lim_{n \to \infty} (\frac{1}{\alpha e^{-k\alpha}[1+\tau_2e^{-\tau_2}-\tau_2]} - k)(\alpha e^{-k\alpha}[1+\tau_2e^{-\tau_2}-\tau_2]) \\&+(k\alpha e^{-k\alpha}-k\alpha\tau_2[e^{-\tau_2k\alpha}-e^{-k\alpha}]) = 1
\end{aligned}
\end{equation}
With the last case, we completed the proof of property $(i)$. Now, for the case $m_1=m_2$,  
\begin{equation}
\begin{aligned}
    \lim_{n \to \infty} \frac{\pi(S_1^{\textbf{P}})}{\pi(S_2^{\textbf{P}})} &= \lim_{n \to \infty} \frac{\pi(S_1^{\textbf{P}})}{\pi(Q_1^{\textbf{P}})}\frac{\pi(Q_2^{\textbf{P}})}{\pi(S_2^{\textbf{P}})}\frac{\pi(Q_1^{\textbf{P}})}{\pi(Q_2^{\textbf{P}})} \\
    &\stackrel{(a)}{=} \lim_{n \to \infty} \frac{\pi(Q_1^{\textbf{P}})}{\pi(Q_2^{\textbf{P}})} \stackrel{(b)}{=} 1
\end{aligned}
\end{equation}
Since the state of the pivot source for both the states $Q_1^{\textbf{P}}$ and $Q_2^{\textbf{P}}$ is $s-1$ and number of active sources is $m_1$ and $m_2$ respectively, (a) follows from property $(i)$ and (b) follows from property $(ii)$. \\
Similarly, for the case $m_1=m_2+1$,
\begin{equation}
\begin{aligned}
    \lim_{n \to \infty} \frac{\pi(S_1^{\textbf{P}})}{\pi(n\,S_2^{\textbf{P}})} &= \lim_{n \to \infty} \frac{\pi(S_1^{\textbf{P}})}{\pi(Q_1^{\textbf{P}})}\frac{\pi(Q_2^{\textbf{P}})}{\pi(S_2^{\textbf{P}})}\frac{\pi(Q_1^{\textbf{P}})}{n\,\pi(Q_2^{\textbf{P}})} \\
    &\stackrel{(a)}{=} \lim_{n \to \infty} \frac{\pi(Q_1^{\textbf{P}})}{n\,\pi(Q_2^{\textbf{P}})} \\
    &\stackrel{(b)}{=} \frac{1}{\alpha{e^{-k\alpha}}+\alpha\tau_2(e^{-\tau_2k\alpha}-e^{-k\alpha})} - k
\end{aligned}
\end{equation}
Since the state of the pivot source for both the states $Q_1^{\textbf{P}}$ and $Q_2^{\textbf{P}}$ is $s-1$ and number of active sources is $m_1$ and $m_2$ respectively, (a) follows from property $(i)$ and (b) follows from property $(iii)$.
\vspace{-0.5cm}
\section*{Acknowledgment}
The work in this paper has been funded in part by TUBITAK grants 117E215, and by Huawei.

\ifCLASSOPTIONcaptionsoff
\newpage
\fi



%
	
	
\bibliography{references} 
%







\end{document}